\documentclass[conference]{IEEEtran}
\usepackage{amssymb, amsthm,bm,mathtools,dsfont} 
\usepackage{enumitem,cite}
\usepackage{color}
\usepackage{comment}
\usepackage[utf8]{inputenc} 
\usepackage[T1]{fontenc}
\usepackage{graphicx} 
\usepackage{empheq}
\usepackage{mleftright} 
\usepackage[nameinlink,capitalize]{cleveref}
\usepackage[dvipsnames]{xcolor}
\usepackage{tikz}
\usepackage{subcaption}
\usetikzlibrary{patterns}
\usepackage{url} 

\newtheorem{theorem}{Theorem}

\newtheorem{proposition}[theorem]{Proposition}
\newtheorem{corollary}[theorem]{Corollary}

\theoremstyle{definition}
\newtheorem{definition}[theorem]{Definition}
\newtheorem{notation}[theorem]{Notation}

\newtheorem{remark}[theorem]{Remark}
\newtheorem{example}[theorem]{Example}

\Crefname{equation}{Eq.}{Eqs.}
\Crefname{figure}{Fig.}{Figs.}
\Crefname{construction}{Constr.}{Constrs.}
\Crefname{example}{Ex.}{Exs.}

\newcommand{\N}{\mathds{N}}

\newcommand{\F}{\mathds{F}}

\newcommand{\R}{\mathds{R}}
\newcommand{\code}{\mathcal{C}}
\newcommand{\Scal}{\mathcal{S}}
\newcommand{\Dcode}{\mathcal{D}}
\newcommand{\Gcal}{\mathcal{G}}
\newcommand{\Gbf}{\mathbf{G}}
\newcommand{\gbf}{\mathbf{g}}
\newcommand{\Bbf}{\mathbf{B}}
\newcommand{\bbf}{\mathbf{b}}
\newcommand{\Fbf}{\mathbf{F}}
\newcommand{\fbf}{\mathbf{f}}
\newcommand{\ebf}{\mathbf{e}}
\newcommand{\Sbf}{\mathbf{S}}

\newcommand{\xbf}{\mathbf{x}}

\newcommand{\PG}{\textnormal{PG}}
\newcommand{\Rcal}{\mathcal{R}}
\newcommand{\Tcal}{\mathcal{T}}

\title{The Oval Strikes Back}
\author{%
  \IEEEauthorblockN{\textbf{Andrea Di Giusto}$^1$, \textbf{Alberto Ravagnani}$^1$ and \textbf{Emina Soljanin}$^2$}
  \IEEEauthorblockA{$^1$Eindhoven University of Technology, the                           Netherlands,   
                    $^2$Rutgers University\\
                    \textit{\{a.di.giusto, a.ravagnani\}@tue.nl}, \textit{emina.soljanin@rutgers.edu}
                    }
}

\begin{document}

\maketitle

\begin{abstract}
    We investigate the applications of ovals in projective planes to distributed storage, with a focus on the Service Rate Region problem. Leveraging the incidence relations between lines and ovals, we describe a class of non-systematic MDS matrices with a large number of small and disjoint recovery sets. For certain parameter choices, the service-rate region of these matrices contains the region of a systematic generator matrix for the same code, yielding better service performance. We further apply our construction to analyze the PIR properties of the considered MDS matrices and present a one-step majority-logic decoding algorithm with strong error-correcting capability. These results highlight how ovals, a classical object in finite geometry, re-emerge as a useful tool in modern coding theory.
\end{abstract}

\section{Introduction}
An \textit{oval} is a set of points in the projective plane $\PG(2,q)$ with the property that no three of them are collinear, and having maximum size.
Many properties of ovals are known since the mid-$20^\text{th}$ century. Their cardinality is~\cite{bose1947mathematical}
\begin{equation}\label{eq:oval_size}
    m(2,q)=\begin{cases}
        q+2\quad\text{if $q$ is even,}\\
        q+1\quad\text{if $q$ is odd,}
    \end{cases}
\end{equation}
and an account of their geometric properties can be found in, e.g.,~\cite[Ch. 8]{hirschfeld1979projective}.
Ovals are mathematically equivalent to linear MDS codes of dimension $k=3$ and maximum length $n=m(2,q)$; see~\cite[Ch. 11]{macwilliams1977theory}.
In this work, we leverage the geometric properties of ovals to find generator matrices of MDS codes that are especially well-suited for distributed storage applications, in particular for optimizing the service performance of the corresponding distributed storage system.

Digital applications generate an ever-increasing amount of data~\cite{rydning2018digitization}, with storage and access efficiency becoming key performance indicators of modern computing systems.
Data objects are usually replicated on multiple servers based on their popularity.
Changes in the request volumes for each data object pose an issue for replication-based storage strategies. Linear codes have been proposed to implement efficient strategies for adding redundancy, thereby increasing robustness and availability; see~\cite{dimakis2011survey} and others.

The \textit{Service Rate Region} (SRR)~\cite{aktacs2017service} has recently emerged as a metric to quantify the ability of a storage system to handle different patterns of requests; see~\cite{alfarano2023service} for an expository paper, and~\cite{aktacs2021service,alfarano2024service} for a detailed introduction.
In the standard model, $k$ data objects are stored across $n\ge k$ servers, and each server can handle requests up to a maximum rate. The SRR is the set of request rates for data objects that the system can support without exceeding each server's serving capacity.
In practice, the redundancy strategy is specified by a generator matrix $\Gbf$ for an $[n,k]_q$ linear block code $\code$.

Formally, the SRR depends on the chosen $\mathbf{G}$, and different generator matrices of the same code can give different regions.
In the literature, the focus is often on systematic matrices, which model storage systems in which an identical copy of every data object is stored on at least one server.
This choice is intuitively motivated by the fact that, at low request rates, systematic replication systems use only the servers storing identical copies and are therefore very efficient.
In contrast, this paper investigates non-systematic systems, in particular for MDS generator matrices. We will argue that they can offer better performance than their systematic counterparts, due to the symmetry in server load balancing. 

\subsubsection*{Related work} This paper is naturally framed in a recently developed line of research, studying the SRR of generator matrices of certain families of codes.
Other works investigated simplex codes~\cite{aktacs2021service,kazemi2020combinatorial}, Hamming codes~\cite{ly2025maximal,choudhary2025service}, binary Reed-Muller codes~\cite{aktacs2021service,kazemi2020geometric,ly2025serviceRM}, and MDS codes~\cite{aktacs2021service,alfarano2024service,ly2025serviceMDS}.
It was also observed in~\cite{kazemi2020combinatorial} how the service rate problem is connected to PIR/batch codes, and these two topics have close ties to majority logic decoding~\cite{fazeli2015codes,ly2025optimum,ly2025maximal}.

\subsubsection*{Our contribution}
We study the SRR of a specific type of non-systematic matrix $\Gbf$ generating an $[n,3]_q$ MDS code of maximum length $n = m(2,q)$.
The algebraic properties of this matrix can be analyzed through the geometric properties of the associated oval in $\PG(2,q)$.
To the best of our knowledge, this is the first work examining the SRR of an MDS code using geometric techniques, which are typically reserved for binary codes such as Reed–Muller codes.

The recovery system we consider differs substantially in its structure from those previously studied for MDS codes~\cite{aktacs2021service,ly2025serviceMDS}; see also~\Cref{rem:difference_other_MDS}.
We characterize the SRR of these matrices, showing that for certain parameter choices it strictly contains the region of a systematic generator matrix of the same code.
This result indicates that non-systematic matrices can offer competitive recovery performance and should not be disregarded \textit{a priori} in favor of systematic ones.

The recovery system we obtain has a simple, elegant description and is well-suited for other applications.
Using our results, we show that the codes considered are $n/2$-PIR codes and admit a 1-Step Majority Logic Decoding algorithm (1S-MLD) that corrects \smash{$\lfloor\frac{n-2}{4}\rfloor$} errors.
This number exceeds the bound of~\cite[Th. 10.2]{peterson1972error}, which is considered a standard for 1S-MLD.
We discuss this discrepancy and highlight new avenues for developing this decoding technique.

The rest of this paper is organized as follows:
In~\Cref{sec:problem_statement} we introduce the SRR and the class of codes we focus on, together with their corresponding geometric properties.
\Cref{sec:main_result} contains the main results of the paper about the non-systematic recovery system and its SRR.
In~\Cref{sec:applications} we explore applications of our findings to Majority-Logic Decoding and PIR codes. \Cref{sec:conclusion} concludes the paper and outlines future work.

\section{The Service Rate Region and Ovals}\label{sec:problem_statement}

We here establish notation and introduce the service-rate region, focusing on its geometric aspects, and recap some geometric properties of ovals.

\noindent\textbf{Notation:}
    $q$ is a power of a prime and $\F_q$ is the finite field of size $q$.
    For $a\in\N$, we let $[a]=\{1,2,\ldots,a\}$.
    $\F_q^a$ is the vector space of $a-$uples over $\F_q$, and $\leqslant$ denotes the subspace relation.
    For $i\in[a]$, we let $\mathbf{e}_i$ denote the $i^{th}$ vector of the standard basis of $\F_q^a$.
    All vectors are column vectors, and the transpose of a matrix/vector is denoted by the superscript~$^T$.

We consider a distributed storage system that redundantly stores $k \ge 1$ objects on $n \ge k$ servers.
The objects are elements of $\F_q$, and each server stores a linear combination of them.
The storage system is thus completely specified by a matrix $\mathbf{G}=(\mathbf{g}_1 \mid\ldots\mid\mathbf{g}_n)\in\F_q^{k\times n}$, where each column $\mathbf{g}_j\in\F_q^k$ contains the coefficients of the linear combination stored by the corresponding servers.
The $k$-tuple of objects $\mathbf{m}=(m_1,\ldots,m_k)\in\F_q^k$ is stored as $\mathbf{c}^T=(c_1,\ldots,c_n)=\mathbf{m}^T\cdot \mathbf{G}$, where $c_j$ is the content of the $j^\text{th}$ server.
Users retrieve each object separately.
To recover an object, a user must download the contents of a subset of servers, with several options available.
For $J\subseteq[n]$, it is readily seen that for any $i\in[k]$, $m_i$ is a linear combination of~$\{c_j\mid j\in J\}$ iff $\mathbf{e}_i\in \langle\mathbf{g}_j\mid j\in J\rangle$, which leads to the following definition.
\begin{definition}\label{def:minimal_rec_set}
    A \textit{recovery set} for the $i^{th}$ object is any set $J\subseteq[n]$ s.t. $\mathbf{e}_i\in \langle\mathbf{g}_j\mid j\in J\rangle$. 
    A recovery set is minimal if it is not properly contained in any other recovery set (for the same object).
    The collection of minimal recovery sets for the~$i^\text{th}$ object is denoted by $\Rcal_i(\mathbf{G})$.
    The \textit{recovery system} of $\mathbf{G}$ is $\Rcal(\mathbf{G})=(\Rcal_1 (\mathbf{G}),\ldots,\Rcal_k (\mathbf{G}))$.
\end{definition}
A more general notion of a recovery system can include non-minimal recovery sets.
However, considering non-minimal sets does not change the SRR~\cite[Proposition 1.8]{alfarano2024service}.
We use the bandwidth model of~\cite[Sec. 1.2]{alfarano2024service}, and assume that each object $m_i$ is requested at a certain rate $\lambda_i\in\R_{\geq 0}$. 
Each server can process requests with a normalized cumulative rate of $\mu=1$.

\begin{definition}\label{def:SRR}
    A vector $(\lambda_1,\ldots,\lambda_k)\subseteq\R^k$ is \textit{supported} by $\mathbf{G}$ if there exists a \textit{feasible allocation} for it, that is, a collection of real numbers $
    \{\lambda_{i,R}\mid i\in[k],R\in\Rcal_i(\mathbf{G})\}$ such that
    \begin{empheq}[left = \empheqlbrace]{alignat=3}
         \lambda_{i,R}\geq0\quad&\mbox{for all $i\in[k]$, $R\in\Rcal_i(\mathbf{G})$,}\\
        \sum_{R\in\Rcal(\mathbf{G})}\lambda_{i,R}=\lambda_i\quad&\mbox{for all $i\in[k]$,}\\
        \sum_{i=1}^k\sum_{\substack{R\in\Rcal_i(G)\\ j\in R}}\lambda_{i,R}\leq 1\quad&\mbox{for all  $j\in[n]$}\label{eq:capacity_bound}.
    \end{empheq}
    The \textit{service rate region} of $\mathbf{G}$, abbreviated as \textit{SRR}, is the set $\Lambda(\mathbf{G})$ of all vectors supported by $\mathbf{G}$.
\end{definition}
 $\Lambda(\Gbf)$ is always a down-monotone polytope,
 often defined by several linear constraints.
Standard simplices are polytopes with a simpler description that are used to express inner and outer bounds for more complex regions~\cite{alfarano2024service}.
\begin{definition}
    For $s\in\R_{\geq0}$ and $k\in\N$, the \textit{standard simplex} of \textit{edge size} $s$ in $\R^k$ is the set
    \begin{equation*}
        \Delta_k(s)=\bigg\{(\lambda_1,\ldots,\lambda_k)\in\R_{\geq0}^k\mid\sum_{i=1}^k\lambda_i\leq s\bigg\}.
    \end{equation*}
\end{definition}
\begin{definition}\label{def:linear_code}
    An $[n,k]_q$ ($[n,k,d]_q$) \textit{code} $\code$ is a $k$-dimensional linear subspace $\code\leqslant\F_q^n$ (with minimum Hamming distance~$d$).
    A code is \textit{Maximum Distance Separable} (\textit{MDS}) if $d+k=n+1$.
    The reader is referred to~\cite{macwilliams1977theory} for the coding theory background.
\end{definition}

For example, if $\Gbf$ generates the $[2^m-1,m,2^{m-1}]_2$ simplex code $\Scal_m$, then $\Lambda(\Gbf)=\Delta_m(2^{m-1})$; see~\cite[Sec. VI.D.1, Theorem 3]{aktacs2021service}.
From the analysis in~\cite{aktacs2021service} it follows that every generator matrix $\Gbf$ of $\Scal_m$ has the same SRR.
We emphasize that this is not true in general: The SRR does depend on the matrix $\mathbf{G}$, and is not a code invariant.
However, the
properties of the code generated by $\mathbf{G}$
are helpful to analyze those of the polytope~$\Lambda(\mathbf{G})$.

In the sequel, we work with the projective space $\PG(k-1,q)$, whose elements are the (\textit{projective}) \textit{points}, and denote by 
\begin{align*}
    \phi:\F_q^k\setminus\{\mathbf{0}\}&\rightarrow\PG(k-1,q)\\
    \xbf=(x_1,\ldots,x_k)^T&\mapsto\phi(\xbf)=\bar{\xbf}=(x_1:\ldots:x_k)^T
\end{align*}
the projectivization map.
For $\bar\xbf\in\PG(k-1,q)$, we write $\xbf=\phi^{-1}(\bar{\xbf})$ to mean that $\xbf\in\F_q^k\setminus\{0\}$ is one of the $q-1$ possible pre-images of $\bar\xbf\in\PG(k-1,q)$ via this map (where the choice is irrelevant).
For every subset $S\subseteq\F_q^k\setminus\{0\}$ we let $\phi(S)=\{\phi(s)\mid s\in S\}$.
The (\textit{projective}) \textit{subspaces} of $\PG_q(k-1)$ are then defined as $\phi(V\setminus\{0\})$, for $V$ a subspace of $\F_q^k$.
The geometric description of a storage scheme is based on the projective set of the matrix~$\Gbf$.
\begin{definition}
    The \textit{projective set} of a matrix $\mathbf{G}=(\mathbf{g}_1 \mid\ldots\mid\mathbf{g}_n)$ is the multiset of points $$\Gcal=\{\bar\gbf_j=\phi(\gbf_j)\mid j\in[n]\}\subseteq\PG_q(k-1).$$
\end{definition}
As mentioned in the introduction, this paper mostly focuses on the SRR of MDS codes.
MDS codes are equivalent to arcs in the projective space~\cite[Ch. 11]{macwilliams1977theory}.
\begin{definition}\label{def:oval}
    An \textit{$n-$arc} in $\PG_q(k-1)$ is a set of $n$ points, no $k$ of which lie on a hyperplane.
    When $k=3$, $n-$arcs of maximum size in $\PG(2,q)$ are called \textit{ovals}.
    
\end{definition}
In particular, the maximum length of an MDS code of dimension $3$ corresponds to the size of an oval in $\PG(2,q)$.
\begin{notation}
    For the rest of the paper, $\code$ is an $[n,3,n-2]_q$ MDS code of maximum length $n=m(2,q)$; see \Cref{eq:oval_size}.
\end{notation}
The projective set of any generator matrix of $\code$ is an oval in $\PG(2,q)$.
Notice that by varying the field size $q$ we obtain codes of different lengths, but the dimension is always $k=3$.
In practice, $\code$ can be assumed to be a suitable extension of a Reed-Solomon code.

We continue by reviewing some geometric facts about ovals in $\PG(2,q)$ from~\cite[Ch. 8]{hirschfeld1979projective}.
Let $\Gcal\subseteq\PG(2,q)$ be an oval, $|\Gcal|=m_q(2)=n$. Then a line in $\PG(2,q)$ is \textit{external}, \textit{tangent} or \textit{secant} to $\Gcal$ if it intersects $\Gcal$ in 0, 1 or 2 points respectively.
The points of $\PG(2,q)$ are classified based on the number of tangents of $\Gcal$ they lie on.
\begin{definition}
    A point $\bar\xbf\in\PG(2,q)$ is \textit{internal}, \textit{on $\Gcal$} or \textit{external} if it lies on 0, 1 or 2 tangents to $\Gcal$.
\end{definition}
Every internal point $\bar\xbf$ lies on $n/2$ secant lines (notice $n$ is even, regardless of the parity of $q$). If $\ell_1$ and $\ell_2$ are distinct secants through $P$, then $\ell_1\cap\Gcal$ and $\ell_2\cap\Gcal$ must be disjoint, since there is a unique line in $\PG(2,q)$ through two points.
These facts imply the following elementary property, which will be key to our construction.
\begin{proposition}\label{prop:secants_incidence}
    Let $\Gcal\subseteq\PG(2,q)$ be an oval, and $\bar\bbf$ be an internal point.
    Let $L(\bar\bbf)$ be the set of secants to $\Gcal$ through $\bar\bbf$. Then $\{\ell\cap\Gcal\mid\ell\in L(\bar\bbf)\}$ is a partition of $\Gcal$ in sets of size 2.
\end{proposition}

\section{Main result}\label{sec:main_result}
The following theorem is the main result of this paper. It describes the non-systematic generator matrix of $\code$ and the associated recovery system.
\begin{theorem}\label{thm:internal_basis}
    Let $q\geq4$ be a prime power, $n=m(2,q)$, and let $\code$ be an $[n,3,n-2]_q$ MDS code.
    There exists a generator matrix~$\Gbf$ of $\code$ with the following properties:
    \begin{enumerate}
        \item $\min\{|R|\,:\,R\in\Rcal_i(\Gbf)\text{ for some }i\in[3]\}=2$;
        \item for every $i\in[3]$, $\Rcal_i^2(\Gbf)=\{R\in\Rcal_i(\Gbf)\mid |R|=2\}$ is a partition of the set $[n]$.
        \end{enumerate}
\end{theorem}
\begin{proof}
    The proof relies on incidence arguments from projective geometry. All geometric results mentioned without proof are from the standard reference~\cite[Ch. 8]{hirschfeld1979projective}.
    Let $\Fbf=(\fbf_1 \mid \ldots \mid \fbf_n)$ be any generator matrix for a $[n,3,n-2]_q$ MDS code, and let $\Gcal=\{\bar\fbf_j=\phi(\fbf_j)\mid j\in[n]\}\subseteq\PG(2,q)$ the associated oval.
    Let $\bar\bbf$ be an internal point and $L(\bar\bbf)$ be the set of secant lines through $\bar\bbf$.
    By~\Cref{prop:secants_incidence}, $\{\ell\cap\Gcal\mid\ell\in L(\bar\bbf)\}$ is a partition of $\Gcal$ in sets of size 2.
    We claim that there exist three linearly independent internal points.
    The proof is divided into two cases: $q$ even and $q$ odd.

    When $q$ is even, every point in $\PG(2,q)\setminus\Gcal$ is internal.
    Since the cardinality of the oval is $m_q(2)=q+2$, it follows that there are $q^2+q+1-(q+2)=q^2-1$ internal points.
    Since, for $q\geq 4$, we have $q^2-1> q+1$, not all of these points lie on a line; hence, 3 independent internal points exist.

    When $q$ is odd, we argue in the same way, with the exception that now we have to account for the presence of external points in $\PG(2,q)\setminus\Gcal$.
    The number of internal points in this case is $q(q-1)/2$, and we have $q(q-1)/2>q+1$ for every odd prime power $q\geq 5$.
    Since $q\geq 4$ by assumption, we deduce the existence of the desired three independent internal points.

    We now turn to the construction of the matrix~$\Gbf$.   
    For every $q$, let $\bar\bbf_1$, $\bar\bbf_2$ and $\bar\bbf_3$ be three linearly independent internal points, $\bbf_i=\phi^{-1}(\bar\bbf_i)$, and let $\Bbf=(\bbf_1 \mid \bbf_2 \mid \bbf_3)\in\F_q^{3\times 3}$.
    Since the points are independent, $\Bbf$ has full rank and $\Gbf=(\gbf_1 \mid \ldots \mid \gbf_n)=\Bbf^{-1}\Fbf$ is another generator matrix for $\code$.
    We now show that $\Gbf$ has the desired properties.
    Regarding 1), since for all $i\in[3]$ we have $\bar\bbf_i\notin\Gcal$, it follows that $\ebf_i\notin<g_j>$ for all $j\in[n]$, showing that the minimal size of a recovery set for any object cannot be 1.
    Hence, proving that 2) holds will also conclude the proof of 1).
    
    By construction, for every $J\subseteq[n]$ and $i \in \{1,2,3\}$, we have that $\ebf_i\in \langle\gbf_j\mid j\in J\rangle$ if and only if $\bbf_i\in\langle\fbf_j\mid j\in J\rangle$, which in turn is equivalent to $\bar\bbf_i\in \langle \bar\fbf_j\mid j\in J \rangle$.
    It follows that $\Rcal_i^2(\Gbf)=\{\{a,b\}\mid\exists\ell\in L(\bar\bbf_i)\textnormal{ s.t. }\ell\cap\Gcal=\{\bar\fbf_a,\bar\fbf_b\}\}$.
    In other words, for every $i \in \{1,2,3\}$, $\Rcal_i^2(\Gbf)$ is given by the partition of $[n]$ in sets of size 2 induced by the incidence relations of secant lines to $\Gcal$ through $\bar\bbf_i$, which concludes the proof.
\end{proof}
\begin{notation}
    For the remainder of the paper, $\Gbf$ is a generator matrix of $\code$ with the properties in~\Cref{thm:internal_basis}.
\end{notation}
\begin{remark}\label{rem:difference_other_MDS}
    The recovery system just described differs from the ones previously studied for MDS codes.
    For a systematic matrix $\Fbf$ with the same parameters as $\Gbf$, one has $\Rcal_i(\Fbf)=\{\{i\}\}\cup\{S\mid S\subseteq[n]\setminus\{i\},|S|=3\}$; see~\cite{aktacs2021service}.
    The non-systematic variants considered in~\cite{ly2025serviceMDS} feature a mix of systematic objects $i$ as above and non-systematic objects $i'$ having $\Rcal_{i'}(\Fbf)=\{S\subseteq[n]\mid|S|=3\}$.
    Hence, the sizes of minimal recovery sets in these constructions are either 1 or 3. This is the main difference with our construction, where the smallest recovery sets for all objects have size 2.
\end{remark}

The description of $\Rcal(\Gbf)$ we just provided allows us to compute the SRR of $\Gbf$ in the following theorem, as well as to establish some of its properties. \Cref{fig:1} shows a plot of $\Lambda(\Gbf)$.

\begin{theorem}\label{thm:internal_basis_region}
    $\Lambda(\mathbf{G})$ is a standard simplex of edge size $n/2$.
    Every point in $\Lambda(\mathbf{G})$ is supported  by an allocation $\{\lambda_{i,R}\mid i\in \{1,2,3\},\, R\in\Rcal_i(\Gbf)\}$ such that $\lambda_{i,R}\neq0$ if and only if $|R|=2$.
\end{theorem}
\begin{proof}
    Since the smallest recovery set has size 2, for every $(\lambda_1,\lambda_2,\lambda_3)\in\Lambda(\Gbf)$ we have $\sum_{i=1}^3\lambda_i\leq n/2$; see~\cite[Lemma~5.1]{alfarano2024service}. Hence $\Lambda(\Gbf)\subseteq\Delta_3(n/2)$.
    Let $\lambda=(\lambda_1,\lambda_2,\lambda_3)\in\Delta_3(n/2)$, and consider the allocation
    \begin{equation*}
        \bigg\{\lambda_{i,R}=\frac{2\lambda_i}{n}\mid i\in[3],R\in\Rcal_i(\Gbf),|R|=2\bigg\}.
    \end{equation*}
    By \Cref{thm:internal_basis}, for every $i\in\{1,2,3\}$ and $j\in[n]$ there exists a unique $R\in\Rcal_i^2(\Gbf)$ such that $j\in R$.
    Hence for every $j\in[n]$~\Cref{eq:capacity_bound} reads
    \begin{equation*}
        \sum_{i=1}^3\sum_{R\in\Rcal_i^2(\Gbf)}\lambda_{i,R}=\sum_{i=1}^3\frac{2\lambda_i}{n}\leq1\,,
    \end{equation*}
    showing that $\lambda$ is supported by $\Gbf$.
    It follows that $\Lambda(\Gbf)\supseteq\Delta_3(n/2)$, and the two polytopes coincide.
    Since the allocation above only uses recovery sets of size 2, this also proves the remaining part of the statement.
\end{proof}

Recall the normalized cost of a request $\lambda$, a metric introduced in~\cite[Sec. VIII]{aktacs2021service}: If $\lambda=(\lambda_1,\ldots,\lambda_k)$ is supported by an allocation $\{\lambda_{i,R}\mid i\in[k],R\in\Rcal_i\}$, the associated cost is
\begin{equation*}
C(\lambda)=\frac{\sum_{i=1}^k\sum_{R\in\Rcal_i}\lambda_{i,R}|R|}{\sum_{i=1}^k\lambda_i}.
\end{equation*}
Using the allocations  in~\Cref{thm:internal_basis_region}, one can check that the normalized cost is constant on $\Lambda(\Gbf)$, in the following sense.
\begin{corollary}
    The normalized service cost of every request vector $\lambda\in\Lambda(\Gbf)$ is $C(\lambda)=2$.
\end{corollary}

\subsection{Comparison with other MDS matrices.}
The SRR of systematic MDS matrices with $n\geq2k$ is known~\cite{aktacs2021service}: for $q\geq5$, we have $n=m(2,q)\geq6=2k$, and for a systematic generator matrix $\Sbf$ for $\code$ we have
\begin{equation}\label{eq:systematic_MDS_region}
    \Lambda(\Sbf)=\{(\lambda_1,\lambda_2,\lambda_3)\in\R_{\geq0}^3\mid\sum_{i=1}^3(\lambda_i+k(1-\lambda_i))\leq n\}.
\end{equation}
It is easy to see that the nonzero vertices of $\Lambda(\Gbf)$ do not satisfy the inequality in~\Cref{eq:systematic_MDS_region}, hence $\Lambda(\Gbf)\nsubseteq\Lambda(\Sbf)$.
\begin{proposition}\label{prop:SRR_containment}
    Let $q\geq11$, then $\Lambda(\Gbf)\supsetneq\Lambda(\Sbf)$.
\end{proposition}
\begin{proof}
    From~\cite[Corollary 6.6]{alfarano2024service} we have that $\Lambda(\Sbf)\subseteq\Delta_3(3+(n-3)/3)$.
    By definition of $n=m(2,q)$, for $q\geq11$ we have $3+(n-3)/3\leq n/2$, implying $\Delta_3(3+(n-3)/3)\subseteq\Delta_3(n/2)=\Lambda(\Gbf)$.
\end{proof}

\begin{figure}
    \centering
    \begin{subfigure}[b]{0.25\textwidth}
        \centering
        \begin{tikzpicture}%
	[x={(1.000000cm, 0.000000cm)},
	y={(0.000000cm, 1.000000cm)},
	z={(-0.200000cm, -0.200000cm)},
	scale=1.000000,
	back/.style={dashed, thin},
	edge/.style={color=orange!95!black, thick},
	facet/.style={fill=orange!95!black,fill opacity=0.500000},
	vertex/.style={inner sep=1pt,circle,draw=red!25!black,fill=red!75!black,thick}]
%
%

\coordinate (0.00000, 0.00000, 0.00000) at (0.00000, 0.00000, 0.00000);
\coordinate (0.00000, 0.00000, 3.00000) at (0.00000, 0.00000, 3.00000);
\coordinate (0.00000, 3.00000, 0.00000) at (0.00000, 3.00000, 0.00000);
\coordinate (3.00000, 0.00000, 0.00000) at (3.00000, 0.00000, 0.00000);
\draw[edge,back] (0.00000, 0.00000, 0.00000) -- (0.00000, 0.00000, 3.00000);
\draw[edge,back] (0.00000, 0.00000, 0.00000) -- (0.00000, 3.00000, 0.00000);
\draw[edge,back] (0.00000, 0.00000, 0.00000) -- (3.00000, 0.00000, 0.00000);
\node[vertex] at (0.00000, 0.00000, 0.00000)     {};
\fill[style={fill=orange!95!black,fill opacity=0.500000}] (3.00000, 0.00000, 0.00000) -- (0.00000, 0.00000, 3.00000) -- (0.00000, 3.00000, 0.00000) -- cycle {};
\draw[edge] (0.00000, 0.00000, 3.00000) -- (0.00000, 3.00000, 0.00000);
\draw[edge] (0.00000, 0.00000, 3.00000) -- (3.00000, 0.00000, 0.00000);
\draw[edge] (0.00000, 3.00000, 0.00000) -- (3.00000, 0.00000, 0.00000);
\node[vertex] at (0.00000, 0.00000, 3.00000)     {};
\node[vertex] at (0.00000, 3.00000, 0.00000)     {};
\node[vertex] at (3.00000, 0.00000, 0.00000)     {};
\end{tikzpicture}
        \caption{Non-systematic region $\Lambda(\Gbf)$.}\label{fig:1}
    \end{subfigure}%
    \begin{subfigure}[b]{0.25\textwidth}
        \centering
        \begin{tikzpicture}%
	[
    x={(1.000000cm, 0.000000cm)},
	y={(0.000000cm, 1.000000cm)},
	z={(-0.200000cm, -0.200000cm)},
	scale=1.000000,
	back/.style={dashed, thin},
	edge/.style={thick},
	facet/.style={fill=blue!95!black,fill opacity=0.500000},
	vertex/.style={inner sep=1pt,circle,draw=red!25!black,fill=red!75!black,thick}
    ]
%
%

\coordinate (0.00000, 0.00000, 0.00000) at (0.00000, 0.00000, 0.00000);
\coordinate (0.00000, 0.00000, 3.00000) at (0.00000, 0.00000, 3.00000);
\coordinate (0.00000, 3.00000, 0.00000) at (0.00000, 3.00000, 0.00000);
\coordinate (3.00000, 0.00000, 0.00000) at (3.00000, 0.00000, 0.00000);

\coordinate (0.00000, 0.00000, 2.66667) at (0.00000, 0.00000, 2.66667);
\coordinate (0.00000, 1.00000, 2.33333) at (0.00000, 1.00000, 2.33333);
\coordinate (0.00000, 2.33333, 1.00000) at (0.00000, 2.33333, 1.00000);
\coordinate (0.00000, 2.66667, 0.00000) at (0.00000, 2.66667, 0.00000);
\coordinate (1.00000, 1.00000, 2.00000) at (1.00000, 1.00000, 2.00000);
\coordinate (1.00000, 2.00000, 1.00000) at (1.00000, 2.00000, 1.00000);
\coordinate (2.00000, 1.00000, 1.00000) at (2.00000, 1.00000, 1.00000);
\coordinate (1.00000, 0.00000, 2.33333) at (1.00000, 0.00000, 2.33333);
\coordinate (1.00000, 2.33333, 0.00000) at (1.00000, 2.33333, 0.00000);
\coordinate (2.33333, 0.00000, 1.00000) at (2.33333, 0.00000, 1.00000);
\coordinate (2.33333, 1.00000, 0.00000) at (2.33333, 1.00000, 0.00000);
\coordinate (2.66667, 0.00000, 0.00000) at (2.66667, 0.00000, 0.00000);
\draw[edge,back] (0.00000, 0.00000, 0.00000) -- (0.00000, 0.00000, 3.00000);
\draw[edge,back] (0.00000, 0.00000, 0.00000) -- (0.00000, 3.00000, 0.00000);
\draw[edge,back] (0.00000, 0.00000, 0.00000) -- (3.00000, 0.00000, 0.00000);
\node[vertex] at (0.00000, 0.00000, 0.00000)     {};
\fill[style={fill=orange!95!black,fill opacity=0.500000}] (3.00000, 0.00000, 0.00000) -- (0.00000, 0.00000, 3.00000) -- (0.00000, 3.00000, 0.00000) -- cycle {};
\draw[edge] (0.00000, 0.00000, 3.00000) -- (0.00000, 3.00000, 0.00000);
\draw[edge] (0.00000, 0.00000, 3.00000) -- (3.00000, 0.00000, 0.00000);
\draw[edge] (0.00000, 3.00000, 0.00000) -- (3.00000, 0.00000, 0.00000);
\node[vertex] at (0.00000, 0.00000, 3.00000)     {};
\node[vertex] at (0.00000, 3.00000, 0.00000)     {};
\node[vertex] at (3.00000, 0.00000, 0.00000)     {};
\fill[facet] (2.66667, 0.00000, 0.00000) -- (2.33333, 0.00000, 1.00000) -- (2.00000, 1.00000, 1.00000) -- (2.33333, 1.00000, 0.00000) -- cycle {};
\fill[facet] (1.00000, 2.33333, 0.00000) -- (0.00000, 2.66667, 0.00000) -- (0.00000, 2.33333, 1.00000) -- (1.00000, 2.00000, 1.00000) -- cycle {};
\fill[facet] (2.33333, 1.00000, 0.00000) -- (2.00000, 1.00000, 1.00000) -- (1.00000, 2.00000, 1.00000) -- (1.00000, 2.33333, 0.00000) -- cycle {};
\fill[facet] (2.00000, 1.00000, 1.00000) -- (1.00000, 1.00000, 2.00000) -- (1.00000, 2.00000, 1.00000) -- cycle {};
\fill[facet] (1.00000, 0.00000, 2.33333) -- (0.00000, 0.00000, 2.66667) -- (0.00000, 1.00000, 2.33333) -- (1.00000, 1.00000, 2.00000) -- cycle {};
\fill[facet] (2.33333, 0.00000, 1.00000) -- (2.00000, 1.00000, 1.00000) -- (1.00000, 1.00000, 2.00000) -- (1.00000, 0.00000, 2.33333) -- cycle {};
\fill[facet] (1.00000, 2.00000, 1.00000) -- (0.00000, 2.33333, 1.00000) -- (0.00000, 1.00000, 2.33333) -- (1.00000, 1.00000, 2.00000) -- cycle {};
\draw[edge] (0.00000, 0.00000, 2.66667) -- (0.00000, 1.00000, 2.33333);
\draw[edge] (0.00000, 0.00000, 2.66667) -- (1.00000, 0.00000, 2.33333);
\draw[edge] (0.00000, 1.00000, 2.33333) -- (0.00000, 2.33333, 1.00000);
\draw[edge] (0.00000, 1.00000, 2.33333) -- (1.00000, 1.00000, 2.00000);
\draw[edge] (0.00000, 2.33333, 1.00000) -- (0.00000, 2.66667, 0.00000);
\draw[edge] (0.00000, 2.33333, 1.00000) -- (1.00000, 2.00000, 1.00000);
\draw[edge] (0.00000, 2.66667, 0.00000) -- (1.00000, 2.33333, 0.00000);
\draw[edge] (1.00000, 1.00000, 2.00000) -- (1.00000, 2.00000, 1.00000);
\draw[edge] (1.00000, 1.00000, 2.00000) -- (2.00000, 1.00000, 1.00000);
\draw[edge] (1.00000, 1.00000, 2.00000) -- (1.00000, 0.00000, 2.33333);
\draw[edge] (1.00000, 2.00000, 1.00000) -- (2.00000, 1.00000, 1.00000);
\draw[edge] (1.00000, 2.00000, 1.00000) -- (1.00000, 2.33333, 0.00000);
\draw[edge] (2.00000, 1.00000, 1.00000) -- (2.33333, 0.00000, 1.00000);
\draw[edge] (2.00000, 1.00000, 1.00000) -- (2.33333, 1.00000, 0.00000);
\draw[edge] (1.00000, 0.00000, 2.33333) -- (2.33333, 0.00000, 1.00000);
\draw[edge] (1.00000, 2.33333, 0.00000) -- (2.33333, 1.00000, 0.00000);
\draw[edge] (2.33333, 0.00000, 1.00000) -- (2.66667, 0.00000, 0.00000);
\draw[edge] (2.33333, 1.00000, 0.00000) -- (2.66667, 0.00000, 0.00000);
\node[vertex] at (0.00000, 0.00000, 2.66667)     {};
\node[vertex] at (0.00000, 1.00000, 2.33333)     {};
\node[vertex] at (0.00000, 2.33333, 1.00000)     {};
\node[vertex] at (0.00000, 2.66667, 0.00000)     {};
\node[vertex] at (1.00000, 1.00000, 2.00000)     {};
\node[vertex] at (1.00000, 2.00000, 1.00000)     {};
\node[vertex] at (2.00000, 1.00000, 1.00000)     {};
\node[vertex] at (1.00000, 0.00000, 2.33333)     {};
\node[vertex] at (1.00000, 2.33333, 0.00000)     {};
\node[vertex] at (2.33333, 0.00000, 1.00000)     {};
\node[vertex] at (2.33333, 1.00000, 0.00000)     {};
\node[vertex] at (2.66667, 0.00000, 0.00000)     {};

\end{tikzpicture}
        \caption{$\Lambda(\Sbf)$ and $\Lambda(\Gbf)$ for $q<11$.}\label{fig:2}
    \end{subfigure}
    \vspace{-0.8cm}
\label{fig:3}
\end{figure}

See~\Cref{fig:2} for a plot comparing
$\Lambda(\Gbf)$ and $\Lambda(\Sbf)$ for $q<11$.
Notice that none of the regions contains the other in this case.
The non-systematic matrices considered in~\cite{ly2025serviceMDS} have an SRR contained in the systematic one with the same parameters.

\subsection{Finding internal points} The explicit construction of the matrix $\Gbf$ in~\Cref{thm:internal_basis} relies on finding internal points.
When $q$ is even, this is an easy task, as all of the points of $\PG(2,q)$ not contained in the oval are internal.
When $q$ is odd instead, $\PG(2,q)\setminus\Gcal$ contains $q(q+1)/2$ external points, and 
testing whether a point is internal or external can be a lengthy procedure when $q$ is large.
If the matrix $\Fbf$ is chosen carefully, one can use the fact that the chosen oval is a conic~\cite{segre1955ovals} to speed up computations.
In fact, let $\alpha\in\F_q$ be a primitive element and
\begin{equation*}
    \Fbf=(\fbf_1\mid\ldots\mid\fbf_n)=\begin{pmatrix}
        1 & 1 & 1 & \ldots & 1 & 0 \\
        0 & \alpha & \alpha^2 & \ldots & \alpha^{q-1} & 0\\
        0 & \alpha^2 & \alpha^4 & \ldots & \alpha^{2(q-1)} & 1
    \end{pmatrix}.
\end{equation*}
The matrix $\Fbf$ generates an extension of a Reed-Solomon code having the required parameters, and $\Fbf$ is usually called the \textit{Vandermonde} generator matrix.
It can be checked that the set $\Gcal=\{\bar\fbf_1,\ldots,\bar\fbf_n\}$ is the conic of equation $y_2^2-y_3y_1=0$.
In this case, internal points are characterized algebraically~\cite[Theorem 8.3.3]{hirschfeld1979projective} as follows. A point $\bar\xbf=(x_1:x_2:x_3)$ is internal if and only if $x_2^2-x_1x_3$ is not a square in $\F_q$.
This provides a quick way to find internal points by sampling random points in $\PG(2,q)$ and applying the algebraic test until 3 independent internal points are found.
We use this property in the following example, which sums up the construction described in this section.
\begin{example}
    Let $q=7$, for which $n=m(2,q)=8$.
    Then $\alpha=3$ is a primitive element and the Vandermonde generator matrix for the $[n,3]_q$ extended Reed-Solomon code is
    \begin{equation*}
        \Fbf=\begin{pmatrix}
            1 & 1 & 1 & 1 & 1 & 1 & 1 & 0\\
            0 & 3 & 2 & 6 & 4 & 5 & 1 & 0\\
            0 & 2 & 4 & 1 & 2 & 4 & 1 & 1
        \end{pmatrix}.
    \end{equation*}
    Using the criterion outlined above, one checks that the points $\bar\bbf_1=(1:0:1)$, $\bar\bbf_2=(1:0:2)$ and $\bar\bbf_3=(1:1:2)$ are internal. Moreover, they are linearly independent.
    Then the matrix $\Gbf$ of~\Cref{thm:internal_basis} is
    \begin{equation*}
        \Gbf=\begin{pmatrix}
            1 & 1 & 1\\
            0 & 0 & 1\\
            4 & 2 & 2
        \end{pmatrix}^{-1}\cdot\Fbf=\begin{pmatrix}
            6 & 0 & 1 & 3 & 0 & 1 & 3 & 4\\
            2 & 5 & 5 & 6 & 4 & 2 & 4 & 3\\
            0 & 3 & 2 & 6 & 4 & 5 & 1 &0
        \end{pmatrix}.   
    \end{equation*}
    The recovery sets of size 2 are
    and it can be checked that
    \begin{align}\label{eq:example_recovery_sets}
        \Rcal_1^2(\Gbf)&=\{\{4,5\},\{3,6\},\{2,7\},\{1,8\}\},\nonumber\\
        \Rcal_2^2(\Gbf)&=\{\{3,4\},\{2,5\},\{6,7\},\{1,8\}\},\\
        \Rcal_3^2(\Gbf)&=\{\{1,3\},\{2,5\},\{4,6\},\{7,8\}\},\nonumber
    \end{align}
    and the SRR then is $\Lambda(\Gbf)=\Delta_3(4)$.
\end{example}

\section{Applications}\label{sec:applications}
In this section, we show how the implications of~\Cref{thm:internal_basis} extend outside of the SRR research stream.

\subsection{PIR codes} Private Information Retrieval (PIR) protocols allow users to access a database without disclosing which item they are downloading to the database owner~\cite{fazeli2015codes}.
In this context, linear codes are used to describe the privacy-preserving access structure.
Essentially, a PIR code is a code for which every message symbol can be recovered from many disjoint recovery sets.
PIR codes are related to the \textit{integral} SRR, that is, the set of points in the SRR that are supported by allocations whose coefficients are all integers~\cite[Proposition 4]{aktacs2021service}.
A simple application of the recovery structure of~\Cref{thm:internal_basis} shows that the use of the matrix $\Gbf$ yields high performance of $\code$ as a PIR code.
More precisely, using the sets in $\Rcal_i^2(\Gbf)$ as a PIR access structure, we have the following
\begin{corollary}
    $\code$ is an $n/2$-PIR code.
\end{corollary}

\subsection{Decoding} Majority-Logic Decoding (MLD) is a broad class of low-complexity error-correcting algorithms that are usually multi-step procedures. For example, the celebrated Reed decoding algorithm~\cite{reed1953class} for binary Reed-Muller codes of order $r$ consists of $r+1$ consecutive steps.
When there is only one step to perform, the algorithm is called 1-Step MLD (1S-MLD).
For more details on MLD algorithms, see~\cite{macwilliams1977theory,peterson1972error}.

It was already observed in the literature that studying the structure of recovery sets of a generator matrix gives useful insights into the design of 1S-MLD algorithms.
For example, a novel 1S-MLD procedure for Reed-Muller codes was recently proposed in~\cite{ly2025optimum}, based on the study of recovery sets conducted in~\cite{ly2025serviceRM}.
In general, each recovery set gives a vote, and errors can flip the votes.
The decoding algorithm can then correct any error pattern that leaves a majority of votes intact for each message symbol.
This mechanism is similar to PIR codes, and, in fact, a connection between them and 1S-MLD has been established in~\cite{fazeli2015codes}.
By carefully selecting the sets used to form the votes, one can design algorithms that correct many error patterns.

The recovery set structure described in~\Cref{thm:internal_basis} can be exploited to correct a large number of errors with 1S-MLD.
In fact, for every message symbol $m_i$, the recovery groups in $\Rcal_i^2(\Gbf)$ give $n/2$ votes with disjoint inputs.
It follows that an error affects exactly one of these groups. As long as the number of errors does not exceed half the number of groups, the majority of the votes will still be correct.
More precisely, we have the following result.

\begin{theorem}\label{thm:1S-MLD_algorithm}
    By encoding $\code$ with the generator matrix $\Gbf$ of~\Cref{thm:internal_basis}, it is possible to correct any \smash{$t\leq\lfloor\frac{n-2}{4}\rfloor$} errors with 1S-MLD. 
\end{theorem}
\begin{proof}
    For every message symbol $m_i$, we use the recovery groups in $\Rcal_i^2(\Gbf)$ to form $n/2$ votes with disjoint inputs.
    With this configuration, any error affects exactly one of these groups for each message symbol.
    Since we need a majority of the votes to be correct, there can be at most \smash{$\lfloor\frac{(n/2)-1}{2}\rfloor$} errors.
\end{proof}
\begin{example}
    Consider again the recovery structure described in~\Cref{eq:example_recovery_sets}, for which \smash{$\lfloor\frac{n-2}{4}\rfloor=1$}.
    One error affects at most one vote for each symbol, leaving three correct votes.
\end{example}

The reader familiar with 1S-MLD may have noticed that the error correcting capability of the algorithm described above is higher than the bound of~\cite[Theorem 10.2]{peterson1972error}, which states that the number $t$ of errors correctable by 1S-MLD is constrained by \smash{$t\leq\lfloor\frac{n-1}{2(d^\perp-1)}\rfloor=\lfloor\frac{n-1}{6}\rfloor$}, where $d^\perp=4$ is the minimum distance of the dual code of $\code$. 
This is not a contradiction, for the bound of~\cite{peterson1972error} applies only to algorithms trying to recover the \textit{codeword symbols}, as is clear from its proof.
Instead, the decoding algorithm of~\Cref{thm:1S-MLD_algorithm} is designed to recover message symbols, for which the recovery structure is subject to different constraints.

We point out this difference because some codes, such as Reed-Solomon codes, are not considered suitable for 1S-MLD in the literature due to~\cite[Theorem 10.2]{peterson1972error}, which gives a discouraging upper bound on the number of correctable errors.
When recovering the message symbols, a more suitable upper bound is the following result, which applies to algorithms operating on disjoint recovery sets.

\begin{theorem}\label{thm:1S-MLD_upper_bound}
    Let $\Fbf$ be the generator matrix of a $[n,k,d]_q$ code $\Dcode$, and let $d^\perp$ denote the minimum distance of $\Dcode^\perp$.
    For $i\in[k]$, let $\Tcal_i\subseteq\Rcal_i(\Gbf)$ have the property that for every $R_1,R_2\in\Tcal_i$, we have $R_1\cap R_2=\emptyset$.
    Let $s=\min\{|R|\mid\exists\, i\in[k]\text{ s.t. }R\in\Tcal_i\}$. Then the number of errors, say $t$, that can be corrected by 1S-MLD using the $\Tcal_i$'s as recovery groups is bounded as follows:
    \begin{equation*}
        t\leq\bigg\lfloor\frac{n-s}{2(d^\perp-s)}\bigg\rfloor.
    \end{equation*}
\end{theorem}
\begin{proof}
    Let $J_1,J_2\in\Tcal_i$. Then $J_1\cup J_2$ contains the support of a codeword of $\Dcode^\perp$.
    Hence $\Rcal_i(\Fbf)$ contains a recovery group of size $s$, and all other groups have size $\geq d^\perp-s$.
    As these are disjoint, one concludes  as for the standard bound in~\cite{peterson1972error}.
\end{proof}

Notice that when $s=1$ we recover~\cite[Theorem 10.2]{peterson1972error}, but in general the bound permits a larger number of correctable errors if $s>1$.
In particular, $s$ can be as big as $d^\perp/2$, and this is the case for the matrix $\Gbf$ of~\Cref{thm:internal_basis}, for which $s=2$.
In fact, the matrix $\Gbf$ achieves the bound of~\Cref{thm:1S-MLD_upper_bound}, implying the algorithm of~\Cref{thm:1S-MLD_algorithm} is optimal.

\section{Future work}\label{sec:conclusion}
The geometric approach developed in this paper can be extended to other classes of MDS matrices, requiring generalizations of the methods introduced here. Our results suggest that such extensions could provide similar benefits, as the same structural ideas can support both storage and related applications, such as decoding. Another promising direction is to explore recovery systems arising from the incidence structures of geometric objects beyond ovals and arcs, potentially revealing new matrices with optimal service performance. 

\section*{Acknowledgment}
A.\ Di Giusto research was supported by the European Commission through grant 101072316. E.~Soljanin's research was partially supported by the National Science Foundation under Grant No. CIF-2122400.

\newpage
\IEEEtriggeratref{11}
\bibliographystyle{plain}
\bibliography{refs}

\end{document}